\documentclass[11pt,letter]{article}

\usepackage{fullpage}
\usepackage[utf8x]{inputenc}
\usepackage{amsmath, amsthm}
\usepackage{wrapfig,graphicx,amssymb,textcomp,array,amsmath}
\usepackage{algpseudocode} 
\usepackage{enumerate}
\usepackage{enumitem}
\usepackage{multirow}
\usepackage{tabularx}
\usepackage{algorithm}
\usepackage{color}
\usepackage{todonotes}
\usepackage[titletoc,title]{appendix}
\usepackage[bookmarks=false, colorlinks = false, hidelinks]{hyperref}

\newcommand{\seqnum}[1]{\href{http://oeis.org/#1}{\underline{#1}}}

\newcommand{\SMQ}{\mathrm{SMQ}}
\newcommand{\PMQ}{\mathrm{PMQ}}
\newcommand{\FindMax}{\mathrm{FindMax}}

\newcommand{\Update}{\mathrm{Update}}
\newcommand{\changed}[1]{{#1}}
\newcommand{\remove}[1]{{}}
\definecolor{mygray}{gray}{0.3}

\newcommand{\etal}{{et al.}}
\newcommand{\inc}{\mathrm{inc}}
\newcommand{\dec}{\mathrm{dec}}

\newcommand{\pred}[2]{\mathrm{pred}(#1,#2)}

\title{Rollercoasters and Caterpillars
}

\author{Therese Biedl\thanks{Cheriton School of Computer Science, University of Waterloo, 
	}
	\and Ahmad Biniaz\footnotemark[1]
	\and Robert Cummings\footnotemark[1]
	\and Anna Lubiw\footnotemark[1]
	\and Florin Manea\thanks{Department of Computer Science, Kiel University.} 
	\and Dirk Nowotka\footnotemark[2] 
	\and Jeffrey Shallit\footnotemark[1]}

\date{\today}
\newtheorem{lemma}{Lemma}

\newtheorem{proposition}{Proposition}
\newtheorem{conjecture}{Conjecture}
\newtheorem{theorem}{Theorem}
\newtheorem{observation}{Observation}
\newtheorem*{problem*}{Problem}

\newtheorem*{invariant*}{Invariant}

\begin{document}
	\maketitle
	\begin{abstract}
		A {\em rollercoaster} is a sequence of real numbers for which every maximal contiguous subsequence, that is increasing or decreasing, has length at least three. By translating this sequence to a set of points in the plane, a rollercoaster can be defined as a polygonal path for which every maximal sub-path, with positive- or negative-slope edges, has at least three points. Given a sequence of distinct real numbers, the rollercoaster problem asks for a maximum-length (not necessarily contiguous) subsequence that is a rollercoaster. It was conjectured that every sequence of $n$ distinct real numbers contains a rollercoaster of length at least $\lceil n/2\rceil$ for $n>7$, while the best known lower bound is $\Omega(n/\log n)$. In this paper we prove this conjecture. Our proof is constructive and implies a linear-time algorithm for computing a rollercoaster of this length. Extending the $O(n\log n)$-time algorithm for computing a longest increasing subsequence, we show how to compute a maximum-length rollercoaster within the same time bound. A maximum-length rollercoaster in a permutation of $\{1,\dots,n\}$ can be computed in $O(n \log \log n)$ time.

		The search for rollercoasters was motivated by orthogeodesic point-set embedding of caterpillars. A {\em caterpillar} is a tree such that deleting the leaves gives a path, called the {\em spine}. A {\em top-view caterpillar} is one of degree 4 such that the two leaves adjacent to each vertex lie on opposite sides of
		the spine. As an application of our result on rollercoasters, we are able to find a planar drawing of every $n$-node top-view caterpillar on every set of $\frac{25}{3}n$ points in the plane, such that each edge is an orthogonal path with one bend. This improves the previous best known upper bound on the number of required points, which is $O(n \log n)$. We also show that such a drawing can be obtained in linear time, provided that the points are given in sorted order.
	\end{abstract}
	
	\section{Introduction}
	A {\em run} in a sequence of real numbers is a maximal contiguous subsequence that is increasing or decreasing. A {\em rollercoaster} is a sequence of real numbers such that every run has length at least three. For example the sequence $(8,5,1,3,4,7,6,2)$ is
	a rollercoaster with runs $(8,5,1)$, $(1,3,4,7)$, $(7,6,2)$, which have lengths $3$, $4$, $3$, respectively. The sequence $(8,5,1,7,6,2,3,4)$ is not a rollercoaster because its run $(1,7)$ has length 2. 
	Given a sequence $S=(s_1,s_2,\dots,s_n)$ of $n$ distinct real numbers, the rollercoaster problem is to find a maximum-size set of indices $i_1<i_2<\dots <i_k$ such that $(s_{i_1},s_{i_2}, \dots,s_{i_k})$ is a rollercoaster. In other words, this problem asks for a longest rollercoaster in $S$, i.e., a longest subsequence of $S$ that is a rollercoaster.

	One can interpret $S$ as a set $P$ of points in the plane by translating each number $s_i\in S$ to a point $p_i=(i,s_i)$. With this translation, a rollercoaster in $S$ translates to a ``rollercoaster'' in $P$, which is a polygonal path whose vertices are points of $P$ and such that every maximal sub-path, with positive- or negative-slope edges, has at least three points. See Figure~\ref{map-fig}(a). Conversely, for any point set in the plane, the $y$-coordinates of the points, ordered by their $x$-coordinates, forms a sequence of numbers. Therefore, any rollercoaster in $P$ translates to a rollercoaster of the same length in $S$.
	
	\begin{figure}[htb]
		\centering
		\setlength{\tabcolsep}{0in}
		$\begin{tabular}{cc}
		\multicolumn{1}{m{.56\columnwidth}}{\centering\includegraphics[width=.48\columnwidth]{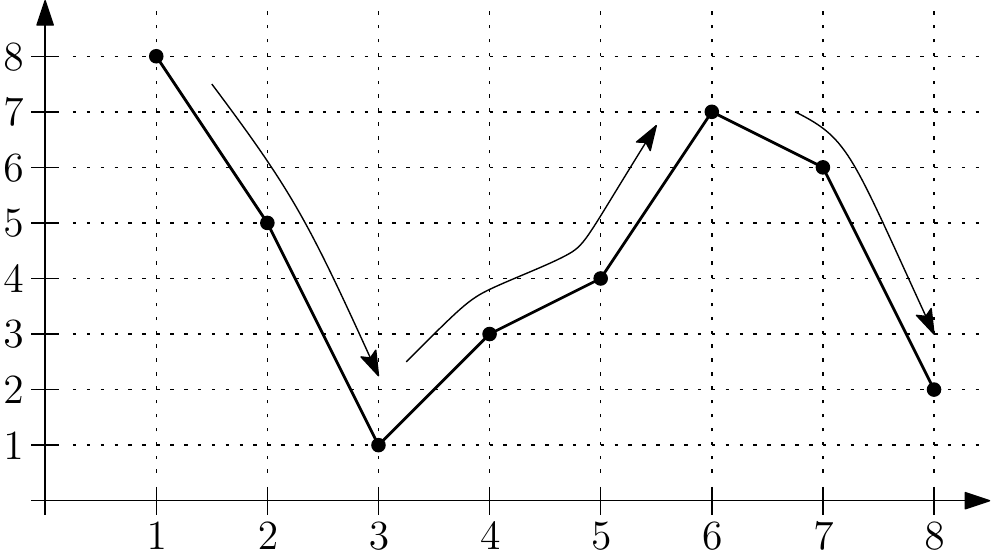}}
		&\multicolumn{1}{m{.44\columnwidth}}{\centering\includegraphics[width=.33\columnwidth]{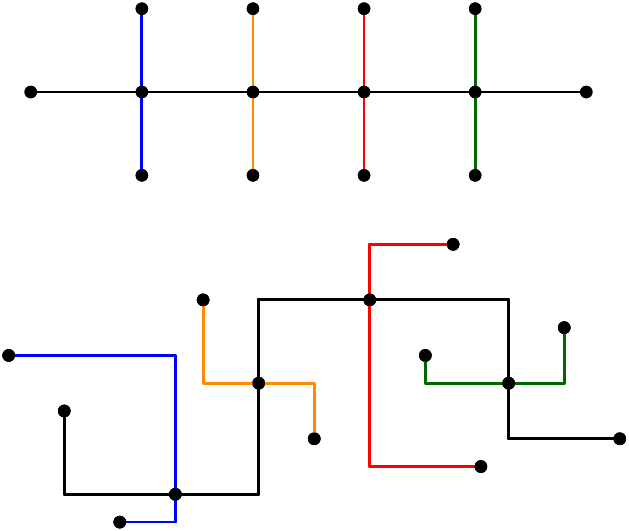}}\\
		(a) & (b)
		\end{tabular}$
		\caption{(a) Translating the sequence $(8,5,1,3,4,7,6,2)$ to a set of points. (b) A planar L-shaped drawing of a top-view caterpillar.}
		\label{map-fig}
	\end{figure}
	
	The best known lower bound on the length of a longest rollercoaster is $\Omega(n/\log n)$ due to Biedl~\etal~\cite{Biedl2017}. They conjectured that 
	
	\begin{conjecture}
		\label{conj1}
		Every sequence of $n>7$ distinct real numbers contains a rollercoaster of length at least $\lceil n/2 \rceil$.
	\end{conjecture}
	
	Conjecture~\ref{conj1} can be viewed as a statement about patterns in permutations,
	a topic with a long history, and the subject of much current research.  
	For example, the Eulerian polynomials, introduced by Euler in 1749,
	are the generating function for the number of descents in permutations.
	For surveys of recent work, see, for example,
	Linton et al.~\cite{Linton:2010} and Kitaev \cite{Kitaev:2011}.  
	Specifically, Conjecture~\ref{conj1} is related to the following seminal result of~Erd\H{o}s and~Szekeres~\cite{Erdos1935} in the sense that they prove the existence of an increasing or a decreasing subsequence of length at least $\sqrt{n}+1$ for $n=ab+1$, which is essentially a rollercoaster with one run.
	\begin{theorem}[Erd\H{o}s and~Szekeres, 1935]
		\label{thr-ES}
		Every sequence of $ab+1$ distinct real numbers contains an increasing subsequence of length at least $a+1$ or a decreasing subsequence of length at least $b+1$.
	\end{theorem}
	
	Hammersley~\cite{Hammersley1972} gave an elegant proof of the Erd\H{o}s-Szekeres theorem that is short, simple, and based on the pigeonhole principle. The Erd\H{o}s-Szekeres theorem also follows from the well-known decomposition of Dilworth (see~\cite{Steele1995}). The following is a restatement of Dilworth's decomposition for sequences of numbers.
	
	\begin{theorem}[Dilworth, 1950]
		\label{thr-Dilworth}
		Any finite sequence $S$ of distinct real numbers can be partitioned into $k$ ascending sequences where $k$ is the maximum length of a descending sequence in $S$.
	\end{theorem}
	
	Besides its inherent interest, the study of rollercoasters is motivated by point-set embedding of caterpillars~\cite{Biedl2017}. A {\em caterpillar} is a tree such that deleting the leaves gives a path, called the {\em spine}. An {\em ordered caterpillar} is a caterpillar in which the cyclic order of edges incident to each vertex is specified. A {\em top-view caterpillar} is an ordered caterpillar where all vertices have degree 4 or 1 such that the two leaves adjacent to each vertex lie on opposite sides of the spine. Planar orthogonal drawings of trees on a fixed set of points in the plane have been explored recently, see e.g., \cite{Biedl2017, Giacomo2013, Scheucher2015}; in these drawings every edge is drawn as an orthogonal path between two points, and the edges are non-intersecting. 
	A {\em planar L-shaped drawing} is a simple type of planar orthogonal drawing in which every edge is an orthogonal path of exactly two segments. Such a path is called an {\em L-shaped edge}. For example see the top-view caterpillar in Figure~\ref{map-fig}(b) together with a planar L-shaped drawing on a given point set.
	Biedl~\etal~\cite{Biedl2017} proved that every top-view caterpillar on $n$ vertices has a planar L-shaped drawing on every set of $O(n \log n)$ points in the plane that is in {\em general orthogonal position}, meaning that no two points have the same $x$- or $y$-coordinate.
	
	\subsection{Our Contributions}
	In Section~\ref{rollercoaster-section} we study rollercoasters and prove Conjecture~\ref{conj1}. In fact we prove something stronger: every sequence of $n$ distinct numbers contains two rollercoasters of total length $n$. Our proof is constructive and yields a linear-time algorithm for computing such rollercoasters. We also extend our result to rollercoasters whose runs are of length at least $k$, for $k>3$. Then we present an $O(n \log n)$-time algorithm for computing a longest rollercoaster, extending the classical algorithm for computing a longest increasing subsequence. This algorithm can be implemented in $O(n\log \log n)$ time if each number in the input sequence is an integer that fits in a constant number of memory words. Then we give an estimate on the number of permutations of $\{1,\dots,n\}$ that are rollercoasters. In Section~\ref{caterpillar-section} we prove, by using Conjecture~\ref{conj1}, that every $n$-node top-view caterpillar has a planar L-shaped drawing on every set of $\frac{25}{3}n$ points in the plane in general orthogonal position.
	
	\section{Rollercoasters}
	\label{rollercoaster-section}
	
	In this section we investigate lower bounds for the length of a longest rollercoaster in a sequence of numbers. We also study algorithmic aspects of computing such rollercoasters. First we prove Conjecture~\ref{conj1}: any sequence of $n$ distinct real numbers contains a rollercoaster of length at least $\lceil n/2\rceil$. Observe that the length 4 sequence $(3,4,1,2)$ has no rollercoaster, so we will restrict to $n\geqslant 5$ in the remainder of this section. Also, due to the following proposition we assume that $n\geqslant 8$. 
	
	\begin{proposition}
		Every sequence of $n\in\{5,6,7\}$ distinct real numbers contains a rollercoaster of length at least $3$. This bound is tight in the worst case.
	\end{proposition}
	
	\begin{wrapfigure}{r}{1.5in} 
		\vspace{-18pt} 
		\centering
		\includegraphics[width=1.2in]{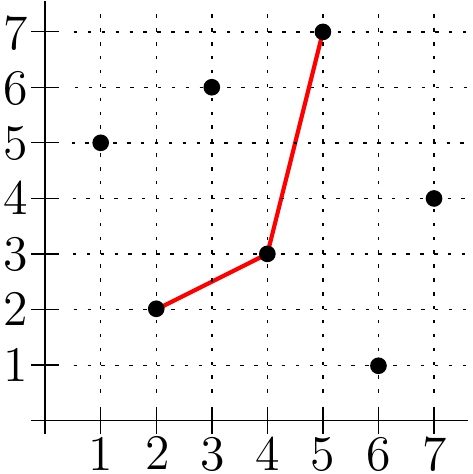} 
		\vspace{-15pt} 
	\end{wrapfigure}
	\noindent{\em{\color{black} Proof.}} By applying Theorem~\ref{thr-ES} with $a=b=2$ we get that every sequence of at least $ab+1=5$ distinct numbers contains an increasing or a decreasing subsequence of length at least $a+1=3$. This \changed{sub}sequence is a rollercoaster of length at least 3.
	For the tightness of this bound, consider the sequence $(5, 2, 6, 3, 7, 1, 4)$, depicted in the adjacent figure. It has length 7 and its longest rollercoaster has length 3.
	\qed \vspace{8pt}

	We refer to a polygonal path as a {\em chain}. We define an {\em ascent} (resp.,~a {\em descent}) as an increasing (resp.,~a decreasing) sequence. We define a {\em k-ascent} (resp.,~a {\em k-descent}) as an ascent (resp.,~a descent) with at least $k$ elements. We also use $k$-ascent and $k$-descent to refer to increasing and decreasing chains with at least $k$ points, respectively. With this definition, a rollercoaster is a sequence in which every run is either a 3-ascent or a 3-descent. We refer to the rightmost run of a rollercoaster as its {\em last run}.

	\subsection{A Proof of Conjecture~\ref{conj1}}
	\label{general-case-section}
	
	In this section we prove the following theorem, which is a restatement of Conjecture~\ref{conj1}. Our proof is constructive, and yields a linear-time algorithm for finding such a rollercoaster.
	
	\begin{theorem}
		\label{rollercoaster3-thr}
		Every sequence of $n\geqslant 8$ distinct real numbers contains a rollercoaster of length at least $\left\lceil n/2\right\rceil$; such a rollercoaster can be computed in linear time. The lower bound of $\left\lceil n/2\right\rceil$ is tight in the worst case.
	\end{theorem}
	
	Consider a sequence with $n \geqslant 8$ distinct real numbers, and let $P$ be its point-set translation with points $p_1,\dots,p_n$ that are ordered from left to right. We define a {\em pseudo-rollercoaster} as a sequence in which every run is a 3-ascent or a 3-descent, except possibly the first run. We present an algorithm that computes two pseudo-rollercoasters $R_1$ and $R_2$ in $P$ such that $|R_1|+|R_2|\geqslant n$; the length of the longer one is at least  $\left\lceil n/2\right\rceil$. Then with a more involved proof we show how to extend this longer pseudo-rollercoaster to obtain a rollercoaster of length at least $\left\lceil n/2\right\rceil$; this will prove the lower bound.
	\subsubsection{An Algorithm}
	\label{algorithm-section}
	First we provide a high-level description of our algorithm as depicted in Figure~\ref{algorithm-fig}. Our algorithm is iterative, and proceeds by sweeping the plane by a vertical line $\ell$ from left to right. We maintain the following invariant:

	\begin{invariant*}
		At the beginning of every iteration we have two pseudo-rollercoasters whose union is the set of all points to the left of $\ell$ and such that the last run of one of them is an ascent and the last run of the other one is a descent.
		\changed{Furthermore, these two last runs have a point in common.}
	\end{invariant*}
	
	\begin{figure}[htb]
		\centering
		\includegraphics[width=.56\columnwidth]{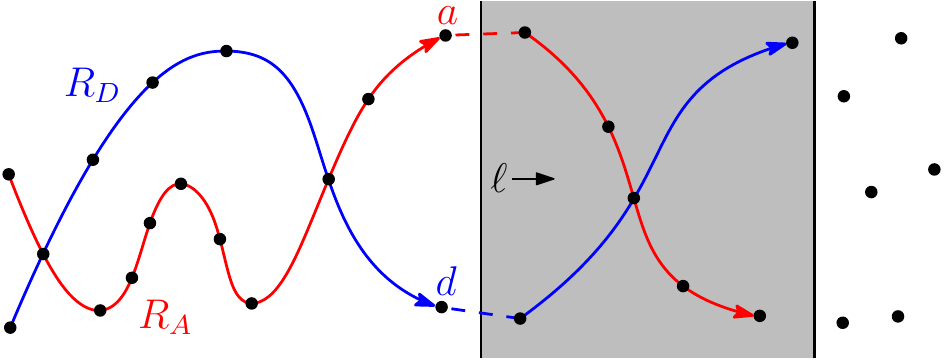}
		\caption{One iteration of algorithm: Constructing two pseudo-rollercoasters.}
		\label{algorithm-fig}
	\end{figure}
	
	During every iteration we move $\ell$ forward and try to extend the current pseudo-rollercoas\-ters. If this is not immediately possible with the next point, then we move $\ell$ farther and stop as soon as we are able to split all the new points into two chains that can be appended to the current pseudo-rollercoasters to obtain two new pseudo-rollercoasters that satisfy the invariant. See Figure~\ref{algorithm-fig}. 
	
	Now we present our iterative algorithm in detail. 
	
	\vspace{5pt}
	{\noindent\bf The First Iteration:} 
	We take the leftmost point $p_1$, and initialize each of the two pseudo-rollercoasters by $p_1$ alone. We may consider one of the pseudo-rollercoasters to end in an ascent and the other pseudo-rollercoaster to end in a descent. 
	\changed{The two runs have a point in common.}
	
	\vspace{5pt}
	{\noindent\bf An Intermediate Iteration:} 
	By the above invariant we have two pseudo-rollercoasters $R_A$ and $R_D$ whose union is the set of all points to the left of $\ell$ and such that the last run of one of them, say $R_A$, is an ascent and the last run of $R_D$ is a descent. 
	\changed{Furthermore, the last run of $R_A$ and the last run of $R_D$ have a point in common.}
	During the current iteration we make sure that every swept point will be added to $R_A$ or $R_D$ or both. We also make sure that at the end of this iteration the invariant will hold for the next iteration. Let $a$ and $d$ denote the rightmost points of $R_A$ and $R_D$, respectively; see Figure~\ref{algorithm-fig}. 
	Let $p_i$ be the first point to the right of $\ell$. If $p_i$ is above $a$, we add $p_i$ to $R_A$ to complete this iteration. Similarly, if $p_i$ is below $d$, we add $p_i$ to $R_D$ to complete this iteration. In either case we get two pseudo-rollercoasters that satisfy the invariant for the next iteration. Thus we may assume that $p_i$ lies below $a$ and above $d$. In particular, this means that $a$ lies above $d$.  
	
	Consider the next point $p_{i+1}$. 
	\changed{(If there is no such point, go to the last iteration.)}
	Suppose without loss of generality that $p_{i+1}$ lies above $p_i$ as depicted in Figure~\ref{two-chain-fig}. Then $d, p_i, p_{i+1}$ forms a 3-ascent. \changed{Continue considering} points $p_{i+2}, \ldots, p_k$ until for the first time, there is a 3-descent in $a, p_i, \ldots, p_k$. 
	In other words, $k$ is the smallest index for which $a, p_i, \ldots, p_k$ contains a descending chain of length 3.
	\changed{(If we run out of points before finding a 3-descent, then go to the last iteration.)} 
	
	\remove{
		Let this descending chain be $p_{k''}, p_{k'}, p_k$ with $p_{k'}$ as far to the right as possible, i.e.,~$k'$ as large as possible, and subject to that, with $k''$ as large as possible. 
		Observe that $k'' \ne i$ since ${i+1}$ would be a better choice. It may happen that $a = p_{k''}$ and $i = {k'}$.
	}
	
	\begin{figure}[htb]
		\centering
		\includegraphics[width=.55\columnwidth]{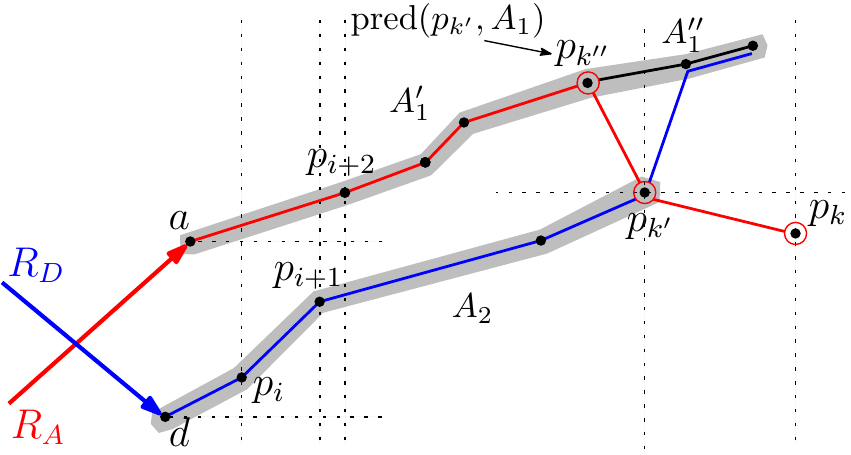}
		\caption{Illustration of an intermediate iteration of the algorithm.}
		\label{two-chain-fig}
	\end{figure}
	
	Without $p_k$ there is no descending chain of length 3. Thus the longest descending chain has two points, and by Theorem~\ref{thr-Dilworth}, the sequence $P'=a, p_i, p_{i+1}, \ldots, p_{k-1}$ is the union of two ascending chains. 
	\changed{We give an algorithm to find two such chains $A_1$ and $A_2$ with $A_1$ starting at $a$ and $A_2$ starting at $p_i$.   
		The algorithm also finds the 3-descent ending with $p_k$.  For every point $q\in A_2$ we define its $A_1$-{\em predecessor} to be the rightmost point of $A_1$ that is to the left of $q$. We denote the $A_1$-predecessor of $q$ by $\pred{q}{A_1}$.
		
		The algorithm is as follows: While moving $\ell$ forward, we denote by $r_1$ and $r_2$ the rightmost points of $A_1$ and $A_2$, respectively; at the beginning $r_1=a$, $r_2=p_i$, and $\pred{p_i}{A_1}=a$. Let $p$ be the next point to be considered. If $p$ is above $r_1$ then we add $p$ to $A_1$. If $p$ is below $r_1$ and above $r_2$, then we add $p$ to $A_2$ and set $\pred{p}{A_1}=r_1$. If $p$ is below $r_2$, then we find our desired first 3-descent formed by (in backwards order) $p_k=p$, $p_{k'}=r_2$, and $p_{k''}=\pred{r_2}{A_1}$.  See Figure~\ref{two-chain-fig}.
		This algorithm runs in time $O(k-i)$, which is proportional to the number of swept points.
		
		We add point $d$ to the start of chain $A_2$.  
		The resulting chains $A_1$ and $A_2$ are shaded in Figure~\ref{two-chain-fig}.
		Observe that $A_2$ ends at $p_{k'}$. 
		Also, all points of $P'$ that are to the right of $p_{k'}$ (if there are any) belong to $A_1$, and  lie to the right of $p_{k''}$, and form an ascending chain. Let $A''_1$ be this ascending chain. Let $A'_1$ be the sub-chain of $A_1$ up to $p_{k''}$; see Figure~\ref{two-chain-fig}. 
		Now we form one pseudo-rollercoaster (shown in red) consisting of  $R_A$ followed by $A'_1$ and then by the descending chain $p_{k''},p_{k'},p_k$.  
		We form another pseudo-rollercoaster (shown in blue) consisting of $R_D$ followed by $A_2$ and then by $A''_1$. 
		We need to verify that the ascending chain added after $d$ has length at least 3.  
		This chain contains $d, p_{i}$ and $p_{k'}$. This gives a chain of length at least 3 unless $k' = i$, but in this case $p_{k''} = a$, so $p_{i+1}$ is part of $A''_1$ and consequently part of this ascending chain.
		Thus we have constructed two longer pseudo-rollercoasters whose union is the set of all points up to point $p_k$, one ending with a 3-ascent and one with a 3-descent and such that the last two runs share the point $p_{k'}$.
		Figure~\ref{first-and-last-fig}(a) shows an intermediate iteration.
	}
	
	\remove{
		Of necessity, one of these starts with $a$ and the other with $p_i$. Let $A_1$ denote the chain that contains $a$ and $A_2$ denote the other one. We can add $d$ to the beginning of $A_2$. These chains are shaded in Figure~\ref{two-chain-fig}. For every point $q\in A_2$ we define its $A_1$-{\em predecessor} to be the rightmost point of $A_1$ that is to the left of $q$. We denote the $A_1$-predecessor of $q$ by $\pred{q}{A_1}$.

		First we show how to compute $A_1$ and $A_2$ in $O(k-i)$ time, which is proportional to the number of swept points. 
		Within the same time we compute the $A_1$-predecessor of every point $q\in A_2$. While moving $\ell$ forward, we denote by $r_1$ and $r_2$ the rightmost points of $A_1$ and $A_2$, respectively; at the beginning $r_1=a$, $r_2=p_i$, and $\pred{p_i}{A_1}=a$. Let $p$ be the next point to be considered. If $p$ is above $r_1$ then we add $p$ to $A_1$. If $p$ is below $r_1$ and above $r_2$, then we add $p$ to $A_2$ and set $\pred{p}{A_1}=r_1$. If $p$ is below $r_2$, then we find our desired first 3-descent with $p_k=p$, $p_{k'}=r_2$, and $p_{k''}=\pred{r_2}{A_1}$. 
		
		In the above construction, $A_2$ ends at $p_{k'}$. Also, all points of $P'$ that are to the right of $p_{k'}$ (if there are any) belong to $A_1$, also lie to the right of $p_{k''}$, and form an ascending chain. Let $A''_1$ be this ascending chain. Let $A'_1$ be the sub-chain of $A_1$ up to $p_{k''}$; see Figure~\ref{two-chain-fig}. Now we form one pseudo-rollercoaster by $R_A$ followed by $A'_1$ and then by the descending chain $p_{k''},p_{k'},p_k$. We obtain another pseudo-rollercoaster by $R_D$ followed by $A_2$ and then by $A''_1$. We need to verify that the ascending chain added after $d$ has length at least 3.  
		This chain contains $d, p_{i}$ and $p_{k''}$ or $p_{k'}$. Since $k'' \ne i$, this gives a chain of length at least 3 unless $k' = i$, but in this case $p_{i+1}$ is part of $A''_1$ and consequently part of this ascending chain.
		Thus we have constructed two longer pseudo-rollercoasters whose union is the set of all points up to point $p_k$, one ending with a 3-ascent and one with a 3-descent.
		Observe that these two pseudo-rollercoasters intersect in $p_{k'}$. Figure~\ref{first-and-last-fig}(a) shows an intermediate iteration.
	}

	\vspace{5pt}
	{\noindent\bf  The Last Iteration:} 
	If there are no points left, then we terminate \changed{the algorithm}. Otherwise, let $p_i$ be the first point to the right of $\ell$. Let $a$ and $d$ be the endpoints of the two pseudo-rollercoasters obtained so far, such that $a$ is the endpoint of an ascent and $d$ is the endpoint of a descent.  
	Notice that $p_i$ is below $a$ and above $d$, because otherwise this iteration would be an intermediate one. For the same reason, the remaining points $p_i,\dots, p_n$ do not contain a 3-ascent together with a 3-descent. If $p_i$ is the last point, i.e., $i=n$, then we discard this point and terminate this iteration. Assume that $i\neq n$, and suppose without loss of generality that the next point $p_{i+1}$ lies above $p_i$. In this setting, by Theorem~\ref{thr-Dilworth} and as described in an intermediate iteration, with the remaining points, we can get two ascending chains $A_1$ and $A_2$ such that $A_2$ contains at least two points. By connecting $A_1$ to $a$ and $A_2$ to $d$ we get two pseudo-rollercoasters whose union is all the points (in this iteration we do not need to maintain the invariant).  
	
	\begin{figure}[htb]
		\centering
		\setlength{\tabcolsep}{0in}
		$\begin{tabular}{cc}
		\multicolumn{1}{m{.35\columnwidth}}{\centering\includegraphics[width=.27\columnwidth]{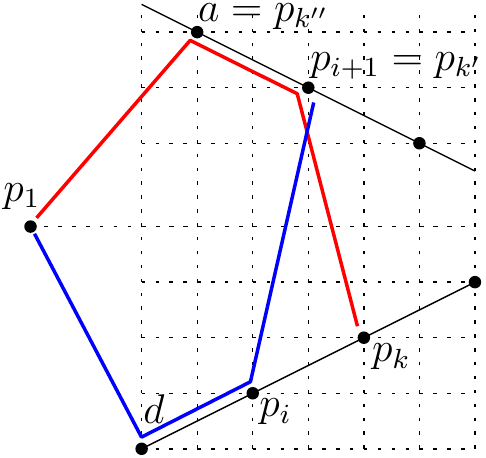}}
		&\multicolumn{1}{m{.65\columnwidth}}{\centering\includegraphics[width=.64\columnwidth]{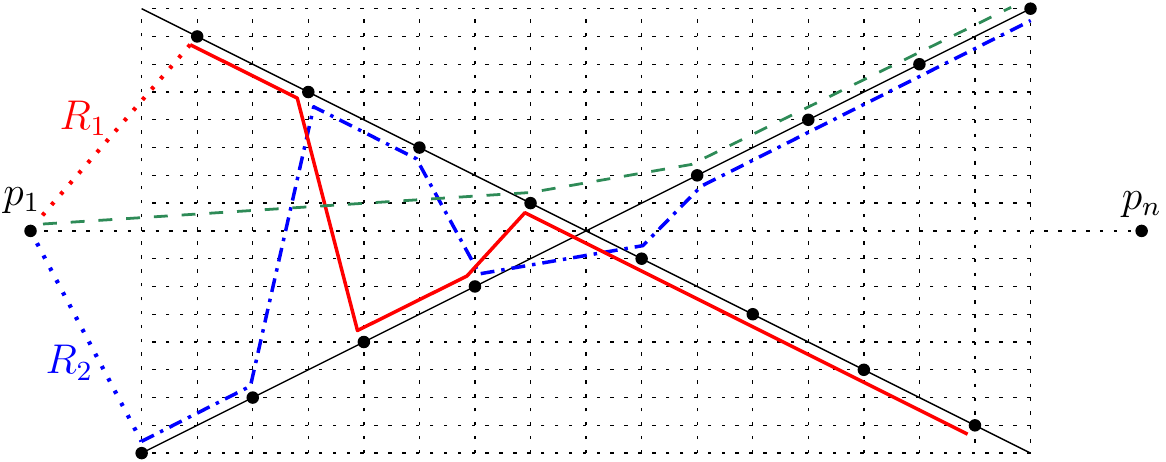}}\\
		(a) & (b)
		\end{tabular}$
		\caption{(a) An intermediate iteration. (b) A point set for which any rollercoaster of length at least $n/4+3$ does not contain $p_1$ and $p_n$. The green (dashed) rollercoaster, which contains $p_1$, has length $n/4+2$. The red (solid) and blue (dash-dotted) chains are the two rollercoasters returned by our algorithm, where blue is a longest possible one.}
		\label{first-and-last-fig}
	\end{figure}

	\vspace{5pt}
	{\noindent\bf  Final Refinement:} 
	At the end of algorithm, we obtain two pseudo-rollercoasters $R_1$ and $R_2$ that share $p_1$ and such that their union contains all points of $P$, except possibly $p_n$. Thus, $|R_1|+|R_2|\geqslant n$, and the length of the longer one is at least $\left\lceil\frac{n}{2}\right\rceil$. 
	
	Recall that every run of pseudo-rollercoasters $R_1$ and $R_2$ is a 3-ascent or a 3-descent, except possibly the first run. If the first run of $R_1$ (resp.,~$R_2$) contains only two points, then we remove $p_1$ to obtain a rollercoaster $\mathcal{R}_1$ (resp.,~$\mathcal{R}_2$). Therefore, we obtain two rollercoasters whose union contains all points, except possibly $p_1$ and $p_n$. 
	
\remove{
		Let $R_1$ be the pseudo-rollercoaster that---in the first iteration---was assumed to have an ascent ending in $p_1$, and let $R_2$ be the one that was assumed to have a descent ending in $p_1$.
		Notice that if the first run of $R_1$ is decreasing, then as described in an intermediate iteration, this run contains at least three points $p_{k''}$, $p_{k'}$, and $p_k$ (with $p_{k''}=a=p_1$). In this case we do not remove $p_1$ from $R_1$. Also if the first run of $R_2$ is increasing, then it contains at least three points $p_1=d$, $p_i$, and $p_{k'}$ or $p_{k''}$ or $p_{i+1}$. In this case we do not remove $p_1$ from $R_2$. Therefore, $p_1$ will be removed from both $R_1$ and $R_2$ if and only if the first run of $R_1$ is an increasing chain with two points and the first run of $R_2$ is a decreasing chain with two points. This configuration is depicted in Figure~\ref{first-and-last-fig}(b).
	}
	
	This is the end of our algorithm. In the next section we analyze the length of the resulting rollercoaster, the tightness of the claimed lower bound, and the running time of the algorithm. 
	
	\subsubsection{Length and Running-Time Analysis}
	
	Our algorithm computes two rollercoasters $\mathcal{R}_1$ and $\mathcal{R}_2$ consisting of all points of $P$, except possibly $p_1$ and $p_n$. Thus, the total length of these rollercoasters is at least $n-2$, and the length of the longer one is at least $\left\lceil\frac{n-2}{2}\right\rceil$. In Appendix~\ref{appA} we improve this bound to $\left\lceil\frac{n}{2}\right\rceil$ by revisiting the first and last iterations of our algorithm with some case analysis. 
	
	We note that there are point sets, with $n$ points, for which every rollercoaster of length at least $n/4+3$ does not contain any of $p_1$ and $p_n$; see e.g., the point set in Figure~\ref{first-and-last-fig}(b). This example shows also the tightness of the $\lceil n/2 \rceil$ lower bound on the length of a longest rollercoaster; the blue rollercoaster is a longest possible and contains $\left\lceil n/2\right\rceil$ points. By removing the point $p_{n-1}$ an example for odd $n$ is obtained.
	
	To verify the running time, notice that the first iteration and final refinement take constant time, and the last iteration is essentially similar to an intermediate iteration. As described in an intermediate iteration the time complexity to find a 3-ascent and a 3-descent for the first time together with the time complexity to compute chains $A'_1$, $A''_1$, and $A_2$ is $O(k-i)$, which is linear in the number of swept points $p_i,\dots,p_k$. Based on this and the fact that every point is considered only in one iteration, our algorithm runs in $O(n)$ time.
	
	\subsection{An Extension}
	\label{extension}
	In this section we extend our result to $k$-rollercoasters. A $k$-{\em rollercoaster} is a sequence of real numbers in which every run is either a $k$-ascent or a $k$-descent. 
	
	\begin{theorem}
		\label{rollercoasterk-thr}
		Let $k\geqslant 4$ be an integer. Then every sequence of $n\geqslant (k-1)^2+1$ distinct real numbers contains a $k$-rollercoaster of length at least $\frac{n}{2(k-1)}-\frac{3k}{2}$.
	\end{theorem}
	\begin{proof}
		Our proof follows the same iterative approach of the proof of Theorem~\ref{rollercoaster3-thr}. Consider a sequence of $n$ distinct real numbers and its point-set translation $p_1,\dots,p_n$.
		We sweep the plane by a line $\ell$, and maintain two $k$-rollercoasters $R_A$ and $R_D$ to the left of $\ell$ such that the last run of $R_A$ is an ascent and the last run of $R_D$ is descent. In each iteration we move $\ell$ forward and stop as soon as we see a $k$-ascent $A$ and a $k$-descent $D$. Then we attach $D$ to $R_A$, and $A$ to $R_D$. To achieve the claimed lower bound, we make sure that the total length of $A$ and $D$ is at least $1/(k-1)$ times the number of swept points. 
		
		Consider an intermediate iteration. Let $m$ be the number of swept points in this iteration and let $P'=(p_i,p_{i+1}\dots,p_{i+m-2},\allowbreak p_{i+m-1})$ be the sequence of these points. Our strategy for stopping $\ell$ ensures that $P'$ contains a $k$-ascent and a $k$-descent, while $P''=(p_i,\dots,p_{i+m-2})$ may contain only one of them but not both. Without loss of generality assume that $P''$ does not contain a $k$-descent. Let $\alpha$ be the integer for which
		\begin{equation}
		\label{eq1}
		(\alpha-1)(k-1)<m-1\leqslant \alpha (k-1).
		\end{equation}
		
		The left-hand side of Inequality~\eqref{eq1} implies that $P''$ has at least $(\alpha-1)(k-1)+1$ points. Having this and our assumption that $P''$ does not contain a $k$-descent, Theorem~\ref{thr-ES} implies that $P''$ contains an increasing subsequence of length at least $\alpha$. We take the longest increasing and the longest decreasing subsequences in $P'$ as $A$ and $D$, respectively. Observe that $|A|\geqslant \max\{k,\alpha\}$ and $|D|=k$. This and the right-hand side of Inequality~\eqref{eq1} imply that 
		\[|A|+|D|\geqslant\alpha+k\geqslant \frac{m-1}{k-1}+k> \frac{m}{k-1},\]
		which means that the total length of $A$ and $D$ is at least $1/(k-1)$ times the number of swept points. 
		In the last iteration if we sweep at most $(k-1)^2$ points then we discard all of them. But if we sweep more than $(k-1)^2$ points then by an argument similar to the one above there exists an integer $\alpha$, with $\alpha\geqslant m/(k-1)$, for which we get either an $\alpha$-ascent or an $\alpha$-descent, which contains at least $1/(k-1)$ fraction of the swept points.
		
		The first iteration is similar to the one in the proof of Theorem~\ref{rollercoaster3-thr}: we assume the existence of an ascent and a descent that end at the first point. At the end of algorithm if the first run of any of $R_A$ and $R_D$ contains $k'$ points, for some $k'<k$, then by removing $k'-1$($\leqslant k-2$) points from its first run we get a valid $k$-rollercoaster. The total length of the resulting two $k$-rollercoasters is
		\[|R_A|+|R_D|\geqslant \frac{n-(k-1)^2}{k-1} - 2(k-2),\]
		where the length of the longer one is at least \[\frac{n-(k-1)^2}{2(k-1)} - (k-2)>\frac{n}{2(k-1)}-\frac{3(k-1)}{2}.\qedhere\]
	\end{proof}
	
	\subsection{Algorithms for a Longest Rollercoaster}
	\label{algorithms}
	
	\begin{wrapfigure}{r}{2in} 
		\vspace{-4pt} 
		\centering
		\includegraphics[width=1.8in]{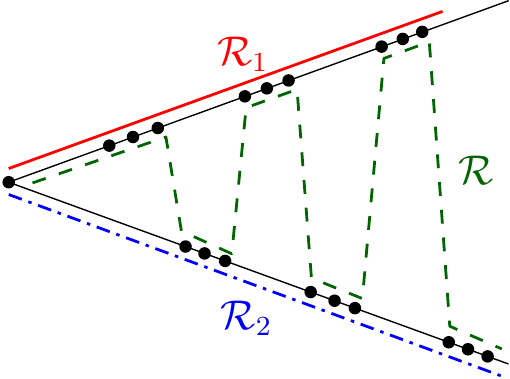} 
		\vspace{-5pt} 
	\end{wrapfigure}
	In this section we study algorithmic aspects of computing a longest rollercoaster in a given sequence $S$ of $n$ distinct real numbers. By Theorem~\ref{rollercoaster3-thr} we can compute a rollercoaster of length at least $\lceil n/2 \rceil$ in $O(n)$ time. However this rollercoaster may not necessarily be a longest one. If we run our algorithm of Section~\ref{algorithm-section} on the point set in the figure to the right, then it returns two rollercoasters $\mathcal{R}_1$ and $\mathcal{R}_2$ each of length at most $\lceil\frac{n}{2}\rceil$ while the longest rollercoaster $\mathcal{R}$ has length $n$. In this section, first we adapt the existing $O(n\log n)$-time algorithm for computing a longest increasing subsequence in $S$ to compute a longest rollercoaster in $S$ within the same time bound. Then we show that if $S$ is a permutation of $\{1,\dots,n\}$, then we can compute a longest rollercoaster in $O(n\log\log n)$ time.  
	
	First we recall Fredman's version of the $O(n\log n)$-time algorithm for computing a longest increasing subsequence~\cite{Fredman75}; for more information about longest increasing subsequence, see Romik \cite{Romik:2015}. We maintain an array $R[i]$, which initially has $R[1]=S[1]$ and is empty otherwise. Then as $i$ proceeds from $2$ to $n$, we find the largest $l$ for which $R[l]<S[i]$, and set $R[l+1]=S[i]$. This insertion ensures that every element $R[l]$ stores the smallest element of $S[1..i]$ in which an increasing subsequence of length $l$ ends. After all elements of $S$ have been processed, the index of the last non-empty element of $R$ is the largest length of an increasing sequence; the corresponding sequence can also be retrieved from $R$. Notice that $R$ is always sorted during the above process. So, the proper location of $S[i]$ in $R$ can be computed in $O(\log n)$ time by a predecessor search, which can be implemented as a binary search. Therefore, this algorithm runs in $O(n\log n)$ time.  
	
	To compute a longest rollercoaster we need to extend this approach. We maintain six arrays $R(w,h)$ with $w\in \{\inc,\dec\}$ and $h\in\{2,3_+, 3_{+}^{\prime}\}$ where $\inc$ stands for ``increasing'', $\dec$ stands for ``decreasing'', and $3_+$ stands for any integer that is at least 3. We define a {\em $w$-$h$-rollercoaster} to be a rollercoaster whose last run has $h$ points and is increasing if $w=\inc$ and decreasing if $w=\dec$. We insert $S[i]$ into arrays $R(\inc,h)$ such that after this insertion the following hold:
	
	\begin{itemize}
		\item $R(\inc,2)[l]$ stores the smallest element of $S[1..i]$ in which an $\inc$-$2$-rollercoaster of length $l$ ends. $R(\dec,2)[l]$ stores the largest element of $S[1..i]$ in which a $\dec$-$2$-rollercoaster of length~$l$ ends.
		\item $R(\inc,3_+)[l]$ stores the smallest element of $S[1..i]$ in which an $\inc$-$3_+$-rollercoaster of length $l$ ends. $R(\dec,3_+)[l]$ stores the largest element of $S[1..i]$ in which an $\dec$-$3_+$-rollercoaster of length $l$ ends.
		\item $R(\inc,3_{+}^\prime)[l]$ stores the largest element of $S[1..i]$ in which an $\inc$-$3_{+}$-rollercoaster of length $l$ ends. $R(\dec,3_{+}^\prime)[l]$ stores the smallest element of $S[1..i]$ in which a $\dec$-$3_{+}$-rollercoaster of length $l$ ends. These arrays will be used when the last run of the current rollercoaster changes from an ascent to a descent, and vice versa. 
	\end{itemize}
	
	We insert $S[i]$ into arrays $R(\dec,h)$ so that analogous constraints hold. To achieve these constraints we insert $S[i]$ as follows:
	
	\begin{itemize}
		\item $R(\inc,2)$: Find the largest index $l$ such that $R(\dec,3_{+}^\prime)[l]<S[i]$. If $S[i]<R(\inc,2)[l+1]$ then update $R(\inc,2)[l+1]=S[i]$.
		
		\item $R(\inc,3_+)$: Find the largest indices $l_1$ and $l_2$ such that $R(\inc,2)[l_1]<S[i]$ and $R(\inc,3_+)[l_2]\allowbreak <S[i]$. Let $l=\max\{l_1,l_2\}$. If $S[i]<R(\inc,3_+)[l+1]$ then update $R(\inc,3_+)[l+1]=S[i]$.
		
		\item $R(\inc,3_{+}^\prime)$: Find the largest index $l_1$ and $l_2$ such that $R(\inc,2)[l_1]\allowbreak <S[i]$ and $R(\inc,3_{+}^\prime)[l_2]\allowbreak <S[i]$. Let $l=\max\{l_1,l_2\}$. If $S[i]>R(\inc,3_{+}^\prime)[l+1]$ then update $R(\inc,3_{+}^\prime)[l+1]=S[i]$.
		
		\item The arrays $R(\dec,h)$ are updated in a similar fashion.
	\end{itemize}
	
	Since our arrays $R(w,h)$ are not necessarily sorted, we cannot perform a predecessor search to find proper locations of $S[i]$. To insert $S[i]$ we need to {\em find} the largest index $l$ such that $R(w,h)[l]$ is smaller (or, alternatively, larger) than $S[i]$ for some $w$ and $h$, and also need to {\em update} contents of these arrays. Thereby, if $A$ is an $R(w,h)$ array, we need to perform the following two operations on $A$:
	
	\begin{itemize}
		\item $\FindMax(A,S[i])$: Find the largest index $l$ such that $A[l]>S[i]$ (or $A[l]<S[i]$). 
		\item $\Update(A,l,S[i])$: Set $A[l]=S[i]$.
	\end{itemize}
	
	We implement each $R(w,h)$ as a Fenwick tree \cite{Fenwick94}, which supports $\FindMax$ and $\Update$ in $O(\log n)$ time. Thus, the total running time of our algorithm is $O(n \log n)$. After all elements of $S$ have been processed, the largest length of a rollercoaster is the largest value $l$ for which $R(w,3_+)[l]$ or $R(w,3_{+}^\prime)[l]$ is not empty; the corresponding rollercoaster can also be retrieved from arrays $R(w,h)$, by keeping the history of the way the elements of these arrays were computed, and then rolling back the computation. 
	\vspace{5pt}
	
	{\noindent\bf A Longest Rollercoaster in Permutations:}
	Here we consider a special case where our input sequence $S$ consists of $n$ distinct integers, each of which can be represented using at most $c$ memory words for some constant $c\geqslant 1$, in a RAM model with logarithmic word size. In linear time, we can sort $S$, using Radix Sort, and then hash it to a permutation of $\{1,\ldots,n\}$. This reduces the problem to finding a longest rollercoaster in a permutation of $\{1,\ldots,n\}$. The longest increasing subsequence of such a sequence can be computed in $O(n \log \log n)$ time by using a van Emde Boas tree \cite{Boas75}, which supports predecessor search and updates in $O(\log \log n)$ time.\footnote{We note that a longest increasing subsequence of a permutation can also be computed in $O(n\log\log k)$ time (see \cite{Crochemore2010}) where $k$ is the largest length of an increasing sequence, which is $\Omega(\sqrt{n})$. However, in our case, the largest length of a rollercoaster is $\Omega(n)$.} To compute a longest rollercoaster in the same time, we need a data structures that supports $\FindMax$ and $\Update$ in permutations in $O(\log\log n)$ time. In Appendix~\ref{appB} we show how to obtain such a data structure by using van Emde Boas trees combined with some other structures.
	
	\begin{lemma}\label{FindMax}
		Let $A$ be an array with $n$ elements from the set $\{0,1,\ldots,n\}$ such that each non-zero number occurs at most once in $A$. We can construct, in linear time, a data structure that performs $\FindMax$ and $\Update$ operations in $O(\log\log n)$ amortized time.
	\end{lemma}
	
	With Lemma~\ref{FindMax} in hand, we can compute a longest rollercoaster in $S$ in $O(n \log\log n)$ time. We note that this algorithm can also compute a longest increasing subsequence by maintaining only the array $R(\inc,3_+)$. 
	
	\subsection{Counting Rollercoaster Permutations}
	
	In this section we estimate the number $r(n)$ of permutations of
	$\{ 1,2, \ldots, n \}$ that are rollercoasters.  A brief table
	follows:
	
	\begin{center}
		\setlength\tabcolsep{4.5pt}
		\begin{tabular}{c|cccccccccccccc}
			$n$ & 1 & 2 & 3 & 4 & 5 & 6 & 7 & 8 & 9 & 10 & 11 & 12 & 13 & 14 \\
			\hline
			$r(n)$ & 1 & 0 & 2 & 2 & 14 & 42 & 244 & 1208 & 7930 & 52710 & 40580 & 3310702 &29742388 & 285103536 
		\end{tabular}
	\end{center}
	
	\noindent This is sequence \seqnum{A277556} in the {\it On-Line Encyclopedia
		of Integer Sequences} \cite{oeis}.
	
	The first step is to rephrase the condition that a permutation is a 
	rollercoaster in the language of ascents and descents.  Given a
	length-$n$ permutation $\pi = \pi_1 \pi_2 \cdots \pi_n$, its 
	{\it descent word} $u(\pi)$ is defined to be $u_1 u_2 \cdots u_{n-1}$
	where $u_i = {\bf a}$ if $\pi_i < \pi_{i+1}$ and ${\bf b}$ otherwise.
	The set of all descent words for rollercoaster permutations is
	therefore given by the expression
	$$ ({\bf bbb}^* + \epsilon)({\bf aaa}^* \, {\bf bbb}^*)({\bf aaa}^* + \epsilon),$$
	which specifies that every increasing run and every decreasing run must contains at least three elements. Since this description is a regular expression,
	one can, in principle,
	obtain the asymptotic behavior of $r(n)$ using the techniques
	of \cite{Basset:2016}, but the calculations appear to be formidable.
	
	Instead, 
	we follow the approach of Ehrenborg and Jung \cite{Ehrenborg&Jung:2012}.
	This is based on specifying sets of permutations through pattern avoidance.
	We say a word $w$ {\it avoids\/} a set of
	words $S$ if no contiguous subword of $w$
	belongs to $S$.
	Although rollercoasters are not specifiable in terms of a finite
	set of avoidable patterns, they ``almost are''.  Consider the
	patterns $\{ {\bf aba}, {\bf bab}\}$.   Every descent word of a
	rollercoaster must avoid both these patterns, and every word avoiding
	these patterns that {\it also\/}
	begins and ends with either $\bf aa$ or $\bf bb$
	is the descent word of some rollercoaster.   Let $s(n)$ be the number
	of permutations of length $n$ whose descent word avoids 
	$\{ {\bf aba}, {\bf bab}\}$.  Then $r(n) = \Theta(s(n))$.
	From \cite[Prop.~5.2]{Ehrenborg&Jung:2012} we know that
	$s(n)\sim c \cdot n! \cdot \lambda^{n-3}$
	where $\lambda \doteq 0.6869765032\cdots$ is the root of a
	certain equation.  It follows that
	$r(n) \sim c' \cdot n! \cdot \lambda^{n-3}$
	where $c'$ is a constant, approximately $0.204$.

\section{Caterpillars}
\label{caterpillar-section}

In this section we study the problem of drawing a top-view caterpillar, with L-shaped edges, on a set of points in the plane that is in general orthogonal position. Recall that a top-view caterpillar is an ordered caterpillar of degree 4 such that the two leaves adjacent to each vertex lie on opposite sides of the spine; see Figure~\ref{map-fig}(b) for an example.
The best known upper bound on the number of required points for a planar L-shaped drawing of every $n$-vertex top-view caterpillar is $O(n\log n)$; this bound is due to Biedl~\etal~\cite{Biedl2017}. 
We use Theorem~\ref{rollercoaster3-thr} and improve this bound to $\frac{25}{3}n\changed{+O(1)}$.

In every planar L-shaped drawing of a top-view caterpillar, every node of the spine, except for the two endpoints, must have its two incident spine edges aligned either horizontally or vertically. Such a drawing of the spine (which is essentially a path) is called a {\em straight-through} drawing. It has been proved in \cite{Biedl2017} that every $n$-vertex path has an $x$-monotone straight-through drawing on any set of at least $c\cdot n \log n$ points, for some constant $c$. The following theorem improves this bound.

\begin{theorem}
	\label{straight-through-thr}
	Any path of $n$ vertices has an $x$-monotone straight-through drawing on any set of at least $3n{-}3$ points in the plane that is in general orthogonal position.
\end{theorem}
\begin{proof}
	Fix an arbitrary set of $3n{-}3$ points.
	As in the proof of Theorem~\ref{rollercoaster3-thr}, 
	find two pseudo-rollercoasters that together cover all but the last point
	and that both contain the first point.
	Append the last point to both sets; we hence obtain two subsequences 
	$R_1,R_2$ with $|R_1|+|R_2|\geq 3n{-}1$ and for which all but the first and
	last run have length at least 3.  
	
	We may assume $|R_1|\geq \tfrac{3}{2}n-\tfrac{1}{2}$, and will find the 
	straight-through
	drawing within it.  To do so, consider any run $r$ of $R_1$ that is neither the
	first nor the last run, and that has even length (hence length at least 4).
	By removing from $r$ one point that is not shared with an adjacent run,
	we turn it into a run of odd length.  Let $R'$ be the subsequence
	that results after applying this to every such run of $R_1$; then $R'$
	satisfies that every run except the first and last one has odd length.  
	Observe that we can find an $x$-monotone straight-through drawing of length 
	$|R'|$ on this, see e.g.~the black path in Figure~\ref{partition-fig} that is 
	drawn on the black points. 
	
	It remains to argue that $|R'|\geq n$.  Let $r_1,\dots,r_\ell$ be the
	runs of $R_1$, and assign to each run $r_i$ all but the last point of $r_i$
	(the last point of $r_i$ is counted with $r_{i+1}$, or not counted at all
	if $i=\ell$).  Therefore
	$|r_1|+\dots+|r_\ell|=|R_1|{-}1\ge \tfrac{3}{2}n-\tfrac{3}{2}$. 
	For each $r_i$ with $2\leq i\leq \ell{-}1$, we remove a point
	only if $|r_i|\geq 3$, hence we keep at least $\tfrac{2}{3}|r_i|$ points.
	Therefore $|R'|\geq |r_1|+|r_\ell|+\sum_{2\leq i\leq \ell{-}1} \tfrac{2}{3}|r_i|
	\geq \sum_{1\leq i\leq \ell} \tfrac{2}{3}|r_i|
	+ \tfrac{1}{3}(|r_1|+|r_\ell|) \geq \tfrac{2}{3} \cdot (\tfrac{3}{2}n-\tfrac{3}{2})+ \tfrac{1}{3}(2+1)=n$ as desired.
	%
\end{proof}

To draw top-view caterpillars, we essentially use 
Theorem~\ref{straight-through-thr} and place the spine on
the resulting straight-through $x$-monotone path. 
But we will get a slightly better factor
if we analyze the number of points directly. 

\begin{theorem}
	\label{caterpillar-thr-new}
	\label{caterpillar-thr}
	Any top-view caterpillar of $n$ vertices has a planar L-shaped drawing on any set of $\tfrac{25}{3}(n{+}4)$ points in the plane that is in general orthogonal position.
\end{theorem}

\begin{wrapfigure}{r}{2.5in} 
	\vspace{-8pt} 
	\centering
	\includegraphics[width=2.3in]{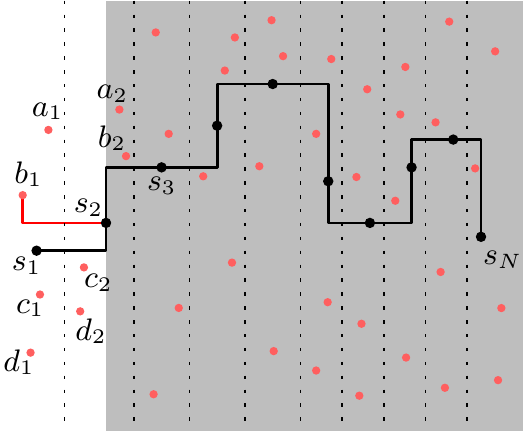} 
	\caption{An $x$-monotone straight-through drawing of an $n'$-vertex path. Red (lighter shade) points are reserved.}
	\label{partition-fig}
	\vspace{-5pt} 
\end{wrapfigure}
\noindent{\em{Proof.}}
Fix any $\tfrac{25}{3}(n{+}4)$ points $P$ in general orthogonal position. 
Partition $P$, by vertical lines, 
into $\tfrac{5}{3}(n{+}4)$ sets, each of them containing five points. We call every such set a {\em 5-set}. Let $P'$ be the set of the mid-points (with respect to 
$y$-coordinates) of every 5-set.  
We have $|P'|\geq \tfrac{5}{3}(n{+}4)$,
so by Theorem~\ref{rollercoaster3-thr} it 
contains a rollercoaster $R$ of length at least $\tfrac{5}{6}(n{+}4)$. 

Let $s_1,\dots,s_N$  (for $N\geq \tfrac{5}{6}(n{+}4)$) be the points of $R$, 
ordered from left to right.
For every $i\in\{1,\dots,N\}$, consider the 5-set containing $s_i$
and let its other points be
$a_i,b_i,c_i,d_i$ from top to bottom; we call these the {\em reserved points}.
The main idea is to draw the spine of the caterpillar along $R$ and the leaves at reserved points, though we will deviate from this
occasionally.
Let the spine consist of vertices $v_1,v_2,\dots,v_\ell$, where $v_1$ and $v_\ell$ are leaves while $v_2,v_3,\dots$ are vertices of degree 4.  
We process the vertices in order along the spine, and maintain the following invariant:

\begin{invariant*}
	At time $k\geq 1$, vertex $v_{2k}$ is drawn on a point $s_i$ that is in the middle of a run of $R$.  Edge $(v_{2k},v_{2k-1})$
	attaches vertically at $v_{2k}$.  All vertices $v_1,\dots,v_{2k-1}$, all their incident leaves, and one incident leaf of $v_{2k}$ are
	drawn on points to the left of $s_i$.
\end{invariant*}

To initiate this process, we draw $v_1$ on $s_1$, $v_2$ on $s_2$, and one leaf incident to $s_2$ on $b_1$.  See 
Figure~\ref{partition-fig}. Clearly the invariant holds for $v_{2}$.   Now assume that vertex $v_{2k}$ has been placed
at $s_i$, and we want to place $v_{2k+1}$ and $v_{2k+2}$ next.  We know that $s_i$ is in the middle of some run of
$R$; up to symmetry we may assume that it is an ascending run.  Let $s_j$ be the last point of this run of $R$;
by the invariant $j>i$.  We distinguish cases:

\medskip\noindent{\bf Case 1:} $j\leq i+4$.
See Figure~\ref{spine1-fig3}(a).   
We will completely ignore the 5-sets containing $s_{i+1},\dots,s_{j-1}$. 
Recall that there are two reserved points $a_j$ and $b_j$ above $s_j$.  We place $v_{2k+1}$ at $b_j$ and $v_{2k+2}$
at $s_{j+1}$.  We connect leaves as follows:  The leaves incident to $v_{2k+1}$ are placed at $a_j$ and $s_j$.
To place one leaf each incident to $v_{2k}$ and $v_{2k+2}$, we use the two points $c_j$ and $d_j$, using the one
farther left for $v_{2k}$.  Clearly the invariant holds.

Observe that there are at most five 5-sets (corresponding to $s_{i+1},\dots,s_{j+1}$) that were parsed, and we have used two for placing spine-vertices.
Therefore, we have used at least $\tfrac{2}{5}$th of the parsed 5-sets.

\medskip\noindent{\bf Case 2:} $j> i+4$.
See Figure~\ref{spine1-fig3}(b). 
We ignore the reserved points corresponding to $s_{i+1}$ and $s_{i+3}$. 
We place $v_{2k+1}$ at $s_{i+2}$ and $v_{2k+2}$ at $s_{i+4}$.  Note that by case-assumption $s_{i+4}$ is {\em not}
the end of the run, so this satisfies the invariant. 
We connect, as leaves, $s_{i+1}$ to $v_{2k}$ (at $s_i$), and $s_{i+3}$ to $v_{2k+2}$ (at $s_{i+4}$). 
The two leaves of $v_{2k+1}$ can be placed in the 5-set of $s_{i+2}$.
We have used four 5-sets and placed two spine-vertices, and have therefore used half of the parsed 5-sets.

\begin{figure}[htb]
	\hspace*{\fill}
	\includegraphics[width=.35\columnwidth,page=1]{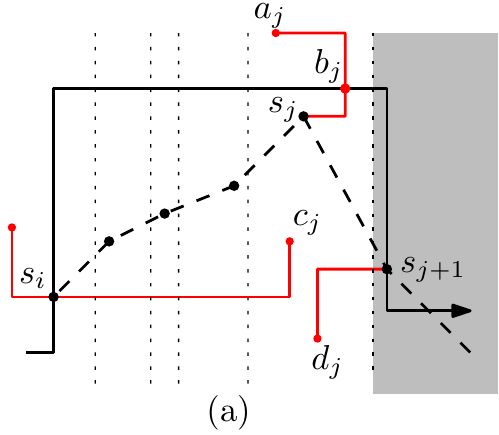}
	\hspace*{\fill}
	\includegraphics[width=.35\columnwidth,page=2]{fig/spine1_therese.pdf}
	\hspace*{\fill}
	\caption{Placing the next two spine-vertices. (a) $j\leq i+4$.  
		(b) $j>i+4$.  The dashed line indicates $R$, the solid line is the spine.}
	\label{spine1-fig3}
\end{figure}

This is the end of one iteration. In every iteration, we have used at least $\tfrac{2}{5}$th of the parsed 5-sets.
Since there were $\tfrac{5}{6}(n{+}4)$ 5-sets, we hence can place $\tfrac{1}{3}(n{+}4)$ spine-vertices.
Since the spine of every $n$-vertex top-view caterpillar has $\tfrac{1}{3}(n{+}4)$ vertices, our claim follows.
\qed \vspace{8pt}

The algorithm in the proof of Theorem~\ref{caterpillar-thr} runs in linear time, provided that the input points are given in sorted order, since we can find the rollercoaster in linear time and then do one scan of the points.


\bibliographystyle{abbrv}
\bibliography{Rollercoaster.bib}

\appendix 
\newpage

\section{Detailed Length Analysis}
\label{appA}
Let $\mathcal{R}_1$ and $\mathcal{R}_2$ be the two rollercoasters computed by our algorithm. Recall that $\mathcal{R}_1\cup \mathcal{R}_2$ contains all points of $P$, except possibly $p_1$ and $p_n$ which might be discarded in the first and last iterations, respectively. Thus the total length of $\mathcal{R}_1$ and $\mathcal{R}_2$ is at least $n-2$, and the length of the longer one is at least $\left\lceil\frac{n-2}{2}\right\rceil$. 
To get the length $\left\lceil\frac{n}{2}\right\rceil$ we revisit the first and last iterations; this includes some case analysis. For brevity, we call a rollercoaster of length at least $\left\lceil\frac{n}{2}\right\rceil$ a {\em suitable rollercoaster}. 
If one of $\mathcal{R}_1$ and $\mathcal{R}_2$ is empty, then the other one is suitable. Assume that none of them is empty, and thus, each contains at least three points (because they are valid rollercoasters). 
Observations~\ref{first-point-obs} and \ref{last-point-obs} follow from the intermediate and last iterations.



\changed{
	\begin{observation}
		\label{first-point-obs}
		If $p_1$ is not in $\mathcal{R}_1\cup \mathcal{R}_2$, then we removed $p_1$ from $R_1$ and $R_2$ in the final refinement at the end of the algorithm.  Thus the first run of $R_1$ is an increasing chain with two points, say $p_1, p_2$ without loss of generality, and the first run of $R_2$ is a decreasing chain with two points, say $p_1, p_3$.  This means that point $p_4$ must lie between $p_2$ and $p_3$ and must start the main case of the intermediate iteration, and this iteration must complete by finding a 3-ascent and 3-descent that share a common point. 
		(In particular, we cannot go directly to the final iteration otherwise one of $R_1$, $R_2$ would be longer.)
		\remove{
			Let $R_1$ be the pseudo-rollercoaster that---in the first iteration---was assumed to have an ascent ending in $p_1$, and let $R_2$ be the one that was assumed to have a descent ending in $p_1$.
			Notice that if the first run of $R_1$ is decreasing, then as described in an intermediate iteration, this run contains at least three points $p_{k''}$, $p_{k'}$, and $p_k$ (with $p_{k''}=a=p_1$). In this case we do not remove $p_1$ from $R_1$. Also if the first run of $R_2$ is increasing, then it contains at least three points $p_1=d$, $p_i$, and $p_{k'}$ or $p_{k''}$ or $p_{i+1}$. In this case we do not remove $p_1$ from $R_2$. Therefore, $p_1$ will be removed from both $R_1$ and $R_2$ if and only if the first run of $R_1$ is an increasing chain with two points and the first run of $R_2$ is a decreasing chain with two points. This configuration is depicted in Figure~\ref{first-and-last-fig}(b).
		}
	\end{observation}
}

\begin{observation}
	\label{last-point-obs}
	\changed{If $p_n$ is not in $\mathcal{R}_1 \cup \mathcal{R}_2$, then the last run of one of them is an ascent ending at a point $u$ and the last run of the other one is a descent ending at a point $v$ such that $u$ lies above $v$, and $p_n$ lies between them.  
	}
	\remove{
		If $p_n$ does not belong to any of $\mathcal{R}_1$ and $\mathcal{R}_2$ then the last run of one of them is an ascent ending at a point $u$ and the last run of the other one is a descent ending at a point $v$ such that $u$ lies above $v$, and $p_n$ lies between them. }
\end{observation}

\changed{Depending on whether $p_1$ is missing from $\mathcal{R}_1\cup \mathcal{R}_2$ or $p_n$ is missing from $\mathcal{R}_1\cup \mathcal{R}_2$, or both, we consider three cases.
	
	\begin{itemize}
		\item $p_1 \not\in \mathcal{R}_1\cup \mathcal{R}_2$, $p_n \in \mathcal{R}_1\cup \mathcal{R}_2$.
		By Observation~\ref{first-point-obs} there is a common point in $\mathcal{R}_1$ and $\mathcal{R}_2$.  
		Then $|\mathcal{R}_1|+|\mathcal{R}_2|\geqslant n$ and we are done.
		
		\item $p_1 \in \mathcal{R}_1\cup \mathcal{R}_2$, $p_n \not\in \mathcal{R}_1\cup \mathcal{R}_2$.
		We apply Observation~\ref{last-point-obs}.  Suppose the last run of  $\mathcal{R}_1$ is a descent and the last run of $\mathcal{R}_2$ is an ascent.  Then at the end of the algorithm the last run of $R_1$ is a descent and the last run of $R_2$ is an ascent and, by the invariant,  they must have a common point.  If that common point is not $p_1$ then $|\mathcal{R}_1|+|\mathcal{R}_2|\geqslant n$ and we are done.  So suppose the common point is $p_1$.   The only danger is that we might remove $p_1$ from one of them, say $R_1$.  But this would imply that the last and only run of $\mathcal{R}_1$ has length less than 3, which is impossible.
		Thus $p_1$ is common to $\mathcal{R}_1$ and $\mathcal{R}_2$, so $|\mathcal{R}_1|+|\mathcal{R}_2|\geqslant n$ and we are done. 
		
		\item $p_1 \not\in \mathcal{R}_1\cup \mathcal{R}_2$, $p_n \not\in \mathcal{R}_1\cup \mathcal{R}_2$.
		By Observation~\ref{first-point-obs} there is a common point $p$ in $\mathcal{R}_1$ and $\mathcal{R}_2$. 
		If there were two or more common points then  $|\mathcal{R}_1|+|\mathcal{R}_2|\geqslant n$ and we would be done.
		Thus we may assume that there is only one common point $p$, which means that the main case of the intermediate iteration only happened once.  
		By Observation~\ref{last-point-obs} the last run of $\mathcal{R}_1$, say,  is a descent ending below $p_n$ and the last run of $\mathcal{R}_2$ is an ascent ending above $p_n$.  Because there is only one common point, these must be the only runs in $\mathcal{R}_1$ and $\mathcal{R}_2$ and we have the situation depicted in Figure~\ref{intersection-fig}.
		
		We have $|\mathcal{R}_1|+|\mathcal{R}_2|\geqslant n-1$, and thus the length of the longer rollercoaster is at least $\left\lceil\frac{n-1}{2}\right\rceil$, which is $\left\lceil\frac{n}{2}\right\rceil$ if $n$ is an even number. Assume that $n$ is an odd number, and thus $n\geqslant 9$. If one of $\mathcal{R}_1$ and $\mathcal{R}_2$ contains less than $\left\lceil\frac{n-1}{2}\right\rceil$ points, then the other one contains at least $\left\lceil\frac{n}{2}\right\rceil$ points. 
		Assume that each of $\mathcal{R}_1$ and $\mathcal{R}_2$ contains $\left\lceil\frac{n-1}{2}\right\rceil$ points, which is at least 4 points. 
		Let $\mathcal{R}'_1$ and $\mathcal{R}'_2$ be the sub-chains of $\mathcal{R}_1$ and $\mathcal{R}_2$, respectively, that are to the left of $p$. Similarly define sub-chains $\mathcal{R}''_1$ and $\mathcal{R}''_2$ to the right of $p$. 
		
		Assume without loss of generality that $|\mathcal{R}'_1| \leqslant |\mathcal{R}'_2|$.
		We consider two cases:
		
		\begin{enumerate}
			\item $|\mathcal{R}'_1| < |\mathcal{R}'_2|$.  Then $|\mathcal{R}''_1| > |\mathcal{R}''_2|$ and $|\mathcal{R}'_2| \geqslant 3$.
			Then the concatenation of $\mathcal{R}'_2$, $p$, and $\mathcal{R}''_1$ is a suitable rollercoaster.
			
			\item $|\mathcal{R}'_1| = |\mathcal{R}'_2|$.Then $|\mathcal{R}''_1| = |\mathcal{R}''_2|$.  By reversing the sequences if necessary, we may assume that $|\mathcal{R}'_1| \geqslant |\mathcal{R}''_1|$, and thus $|\mathcal{R}'_1| \geqslant 3$.  
			Let $\mathcal{R}_1=(a_1,a_2,a_3,a_4,\dots)$ and $\mathcal{R}_2=(b_1,b_2,b_3,b_4,\dots)$. 
			If $a_2$ is to the right of the vertical line through $b_2$, then as depicted in Figure~\ref{intersection-fig} the rollercoaster $(b_1,b_2,a_2,a_3,a_4,\dots)$ is suitable, otherwise $(a_1,a_2,b_2,b_3,b_4,\dots)$ is suitable.
			
		\end{enumerate}
		
	\end{itemize}
}

\remove{
	Depending on whether or not $p_1$ belongs to $\mathcal{R}_1$ or $\mathcal{R}_2$ we consider four cases.
	
	\begin{itemize}
		\item $p_1\in \mathcal{R}_1$ and $p_1\in \mathcal{R}_2$. Since $p_1$ belongs to both $\mathcal{R}_1$ and $\mathcal{R}_2$, and $p_n$ is the only point that may not be in either of $\mathcal{R}_1$ and $\mathcal{R}_2$, we have that $|\mathcal{R}_1|+|\mathcal{R}_2|\geqslant n$, and our result follows.
		\item $p_1\notin \mathcal{R}_1$ and $p_1\notin \mathcal{R}_2$. Recall---from the final refinement---that $\mathcal{R}_1$ is obtained from the pseudo-rollercoaster $R_1$ by removing $p_1$ and $\mathcal{R}_2$ is obtained from the pseudo-rollercoaster $R_2$ by removing $p_1$. Since $p_1$ is removed from both $R_1$ and $R_2$, as described at the end of final refinement, $R_1$ starts with an increasing chain with 2 points followed by a 3-descent, and $R_2$ starts with an decreasing chain with 2 points followed by a 3-ascent. In this setting, Observation~\ref{switch-obs} implies that $\mathcal{R}_1$ and $\mathcal{R}_2$ share a point $p\in\{p_2,\dots,p_{n-1}\}$. To this end, if $p_n\in \mathcal{R}_1\cup \mathcal{R}_2$, then we have $|\mathcal{R}_1|+|\mathcal{R}_2|\geqslant n$, and our result follows. Assume that $p_n\in \mathcal{R}_1\cup \mathcal{R}_2$. Observation~\ref{last-point-obs} and our assumption that $\mathcal{R}_1$ and $\mathcal{R}_2$ share at most one point, imply that the runs of $\mathcal{R}_1$ and $\mathcal{R}_2$ that contain $p$, wont switch from an ascent to a descent nor vice versa. This configuration is depicted in Figure~\ref{intersection-fig}(a).    
		
		To this end we have $|\mathcal{R}_1|+|\mathcal{R}_2|\geqslant n-1$, and thus the length of the longer rollercoaster is at least $\left\lceil\frac{n-1}{2}\right\rceil$, which is $\left\lceil\frac{n}{2}\right\rceil$ if $n$ is an even number. Assume that $n$ is an odd number, and thus $n\geqslant 9$. If one of $\mathcal{R}_1$ and $\mathcal{R}_2$ contains less than $\left\lceil\frac{n-1}{2}\right\rceil$ points, then the other one contains at least $\left\lceil\frac{n}{2}\right\rceil$ points. Assume that each of $\mathcal{R}_1$ and $\mathcal{R}_2$ contains $\left\lceil\frac{n-1}{2}\right\rceil$ points, which is at least 4 points. Let $\mathcal{R}_1=(a_1,a_2,a_3,a_4,\dots)$ and $\mathcal{R}_2=(b_1,b_2,b_3,b_4,\dots)$. Let $\mathcal{R}'_1$ and $\mathcal{R}'_2$ be the sub-chains of $\mathcal{R}_1$ and $\mathcal{R}_2$, respectively, that are to the left of $p$. Similarly define sub-chains $\mathcal{R}''_1$ and $\mathcal{R}''_2$ to the right of $p$. We consider three cases
		
		\begin{enumerate}
			\item $|\mathcal{R}'_1|\geqslant 2$ and $|\mathcal{R}'_2|\geqslant 2$. If $a_2$ is to the right of the vertical line through $b_2$, then as depicted in Figure~\ref{intersection-fig}(a) the rollercoaster $(b_1,b_2,a_2,a_3,a_4,\dots)$ is a suitable, otherwise $(a_1,a_2,b_2,b_3,b_4,\dots)$ is suitable.
			\item $|\mathcal{R}'_1|\geqslant 2$ and $|\mathcal{R}'_2|< 2$. In this case $|\mathcal{R}''_2|\geqslant 2$, moreover $|\mathcal{R}''_2|>|\mathcal{R}''_1|$. Thus, the concatenation of $\mathcal{R}'_1$, $p$, and $\mathcal{R}''_2$ is a suitable rollercoaster.
			\item $|\mathcal{R}'_1|< 2$ and $|\mathcal{R}'_2|< 2$. In this case  $|\mathcal{R}''_1|\geqslant 2$ and $|\mathcal{R}''_2|\geqslant 2$, and we obtain a suitable rollercoaster as in case 1.
		\end{enumerate}
		
		\item $p_1\in \mathcal{R}_1$ and $p_1\notin \mathcal{R}_2$. In this case $|\mathcal{R}_1|+|\mathcal{R}_2|\geqslant n-1$ assuming that $p_n\notin \mathcal{R}_1\cup \mathcal{R}_2$. Recall---from the final refinement---that $\mathcal{R}_1=R_1$ and $\mathcal{R}_2$ is obtained from $R_2$ by removing $p_1$. In this case, $\mathcal{R}_1$ starts with a 3-ascent. Moreover, $R_2$ starts with a decreasing chain with 2 points followed by a 3-ascent; this implies that $\mathcal{R}_2$ starts with a 3-ascent. Since $p_n\notin \mathcal{R}_1\cup \mathcal{R}_2$, Observation~\ref{last-point-obs} implies that the rightmost run of one rollercoaster has to switch from an ascent to a descent. If this switch happens, then by Observation~\ref{switch-obs} either the two rollercoasters share a point (Figure~\ref{intersection-fig}(b)), which implies that $|\mathcal{R}_1|+|\mathcal{R}_2|\geqslant n$, or $p_n$ will be added to $\mathcal{R}_1$ or $\mathcal{R}_2$, which contradicts our assumption that $p_n\notin \mathcal{R}_1\cup \mathcal{R}_2$.
		\item $p_1\notin \mathcal{R}_1$ and $p_1\in \mathcal{R}_2$. This case can be handled similar to the previous case.
	\end{itemize}
}

\begin{figure}[htb]
	\centering
	\includegraphics[width=.33\linewidth]{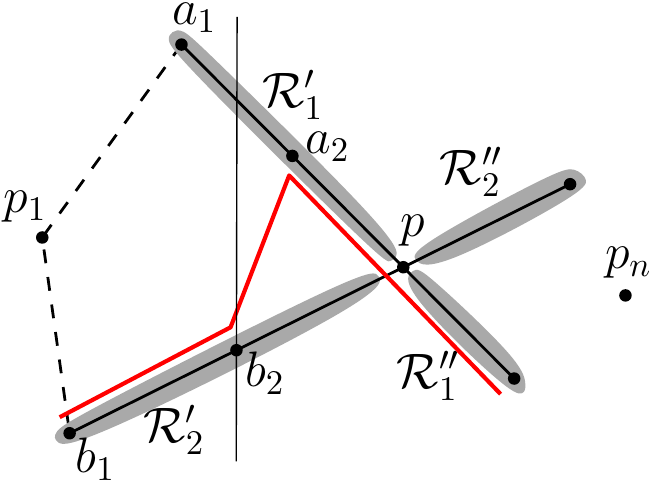}
	\caption{Illustration of the case where $p_1, p_n \notin \mathcal{R}_1  \cup \mathcal{R}_2$.}
	\label{intersection-fig}
\end{figure}

\section{Data Structure for FindMax and Update}
\label{appB}

Recall from Section~\ref{algorithms} that we have an array $A$ containing distinct elements from $\{1,\dots,n\}$, and we want to perform the following two operations on $A$ in $O(\log\log n)$ time:

\begin{itemize}
	\item $\FindMax(A,x)$: Find the largest index $l$ such that $A[l]>x$. 
	\item $\Update(A,l,x)$: Set $A[l]=x$.
\end{itemize}

\begin{lemma}\label{smq} 
	Let $B$ be an array with $n$ elements from the set $\{0,1,\ldots,n\}$ such that each non-zero number occurs at most once in $B$. We can construct, in linear time, a data structure that, 
	\changed{in $O(\log \log n)$ amortized time,  performs the operations Suffix Maximum Query (SMQ) and Update defined as follows:}
	
	\begin{itemize}
		\item $\SMQ(l)$: Return the position of the largest element in $B[l..n]$, and
		\item $\Update(B,l,x)$: Set $B[l]=x$, assuming that 
		$x \in \{1,\ldots,n\}$ is not currently in $B$.
	\end{itemize}
\end{lemma}

\begin{proof}
	In linear time, we construct a doubly linked list $L$ (by indices of $B$) as follows. Set $M=n$ and initialize $L$ by $M$. For $i$ from $n-1$ to $1$ we do the following: if $B[i]>B[M]$, we add the index $i$ to the front of $L$ and then set $M=i$. 
	After this, $L$ contains indices $i_1,i_2,\ldots,i_k$, where $i_1<i_2<\dots<i_k=n$, and 
	\[ B[i_j] =\text{largest element in}
	\begin{cases}
	B[1 .. n]       &     \quad \text{for ~~} j=1\\
	B[(i_{j-1}+1) .. n]       & \quad \text{for all ~~} 1< j\leqslant k\\
	\end{cases}
	\]
	This means that every position that will be returned by $\SMQ$ queries on $B$ belongs to $L$. 
	More precisely, for every $l\in\{1,\dots,n\}$ we have $\SMQ(l)=i_j$ where $i_j$ is the smallest index in $L$ that is not smaller than $l$.
	
	Then in linear time we construct a van Emde Boas tree $T$ with the elements of $L$; this tree allows us to perform predecessor search, successor search, insert, and delete in $O(\log \log n)$ time. 
	To answer an $\SMQ(l)$ query, we report the smallest element $i_j$ in $L$ such that $i_j\geqslant l$. This can be done in $O(\log \log n)$ time by reporting the successor of $l$ in $T$. 
	
	We implement the $\Update(B,l,x)$ operation as follows. Notice that this operation will not change the answer to $\SMQ(t)$ queries for any $t>l$. Thus, we only need to update $L$ and $T$ such that $\SMQ(t)$ reports a correct position when $t\leqslant l$. We find in $L$ the largest index $i_h$ such that $i_h \leqslant l$; this can be done is done in $O(\log \log n)$ time by a predecessor search in $T$. If $B[i_{h+1}]>x$ then we don't do anything as the position of $x$ will not be reported by any $\SMQ$ query. Assume that $B[i_{h+1}]<t$. Then we insert $l$ in $L$, between $i_h$ and $i_{h+1}$, and also in $T$. Then we remove from $L$ and $T$ every index $i_j$, with $j\leqslant h$, for which $B[i_j]<x$. This can be done in $O(r \log \log n)$ where $r$ is the number of elements we that remove. Since every index of $B$ is inserted at most once in $L$, and also deleted once, the amortized running time of this operation is $O(\log \log n)$.
\end{proof}

A data structure similar to that of Lemma~\ref{smq} can be obtained for Prefix Maximum Queries $\PMQ(l)$, which ask for the position of the largest element inf $B[1..l]$.
Now we prove Lemma~\ref{FindMax}, which is restated below.

\vspace{10pt}
\noindent{\bf Lemma~\ref{FindMax}.} {\em 
	Let $A$ be an array with $n$ elements from the set $\{0,1,\ldots,n\}$ such that each non-zero number occurs at most once in $A$. We can construct, in linear time, a data structure that performs $\FindMax$ and $\Update$ operations in $O(\log\log n)$ amortized time.
}

\begin{proof}
	We construct a van Emde Boas tree $T$ for $A$ (we maintain links between every element of $T$ and its occurrence in $A$). We construct an array $B$ where $B[i]=j$ if and only if $A[j]=i$. Then we preprocess $B$ as in the proof of Lemma \ref{smq}.
	
	To answer a $\FindMax(A,x)$ query we proceed as follows.  We describe the case where we look for the largest index $l$ such that $A[l]>x$; the description of the case where $A[l]<x$ is analogous. We first compute the successor $A[j]$ of $x$ in $A$, i.e., the smallest element in $A$ that is greater than $x$. This can be done by a successor search in $T$. Since $A[j]\geqslant x$, we have $\FindMax(A,x)\geqslant j$ and $A[\FindMax(A,x)]\geqslant A[j]$. Hence, to retrieve $\FindMax(A,x)$ we need to find the largest index $k$, with $k\geqslant j$, such that $A[k]\geqslant A[j]$. This index $k$ is the maximum element in $B[A[j]..n]$, which is $\SMQ(B,A[j])$. By Lemma~\ref{smq} this can be done in $O(\log \log n)$ time. Therefore, in $O(\log \log n)$ time, we can answer $\FindMax(A,x)$.
	
	To perform $\Update(A,l,x)$, we set $A[l]=x$, update $T$ (by deleting the old value $A[l]$ and inserting $x$), and execute $\Update(B,x,l)$. These operations take $O(\log \log n)$ amortized time. 
\end{proof}

\end{document}